\documentclass[12pt,reqno]{article}

\usepackage[usenames]{color}
\usepackage{amsmath,amssymb,amsthm}
\usepackage{url}
\usepackage{MnSymbol}

\usepackage[colorlinks=true,
linkcolor=webgreen,
filecolor=webbrown,
citecolor=webgreen]{hyperref}

\definecolor{webgreen}{rgb}{0,.5,0}
\definecolor{webbrown}{rgb}{.6,0,0}

\theoremstyle{plain}
\newtheorem{theorem}{Theorem}

\theoremstyle{definition}

\newcommand{\seqnum}[1]{\href{http://oeis.org/#1}{\underline{#1}}}

\newcommand{\beq}{\begin{equation}}
\newcommand{\eeq}{\end{equation}}

\renewcommand{\P}{\mathcal{P}_{\text{fin}}}
\newcommand{\X}{\mathcal{X}}

\DeclareMathOperator{\coeff}{\mathrm{Coeff}}
\newcommand\mycom[2]{\genfrac{}{}{0pt}{}{#1}{#2}}

\begin{document}

\begin{center}

\vskip 1cm{\LARGE\bf 
Computing the inverses,
their power sums, and extrema
for Euler's totient 
and other multiplicative functions\\
\vskip .1in}

\vskip 1cm

\large

Max A. Alekseyev\\
Department of Mathematics \\
The George Washington University \\
801 22nd St. NW, Washington, DC 20052 \\
United States \\
\href{mailto:maxal@gwu.edu}{\tt maxal@gwu.edu}
\end{center}

\vskip .2 in

\begin{abstract}
We propose a generic algorithm for computing the inverses of a multiplicative function under the assumption that the set of inverses is finite.
More generally, our algorithm can compute certain functions of the inverses, such as their power sums (e.g., cardinality) or extrema, without direct enumeration of the inverses.
We illustrate our algorithm with Euler's totient function $\varphi(\cdot)$ and the $k$-th power sum of divisors $\sigma_k(\cdot)$. 
For example, we can establish that the number of solutions to $\sigma_1(x) = 10^{1000}$ is $15,512,215,160,488,452,125,793,724,066,873,737,608,071,476$, 
while it is intractable to iterate over the actual solutions.
\end{abstract}

\section{Introduction}

A value of a multiplicative function $f$ on a positive integer $n$ equals the product of its values on the prime powers in the prime factorization of $n$.
That is, if $n=p_1^{e_1}\cdot p_2^{e_2}\cdots p_m^{e_m}$, where $p_1<p_2<\dots<p_m$ are primes and $e_1,e_2,\dots,e_m$ are positive integers, then
$$f(n) = \prod_{i=1}^m f(p_i^{e_i}).$$
In particular, $f(1)=1$.
Famous examples of multiplicative functions include $\tau(n)$, the number of divisors of $n$ (with $\tau(p^e)=e+1$);
$\sigma_k(n)$, the $k$-th power sum of divisors of $n$ (with $\sigma_k(p^e) = \frac{p^{k(e+1)}-1}{p^k-1}$); 
and Euler's totient function $\varphi(n)$ (with $\varphi(p^e)=(p-1)\cdot p^{e-1}$).

In the present work, we propose a generic algorithm for computing the set of inverses (full pre-image) $f^{-1}(n)$ of a multiplicative function $f$ for a given integer $n$
under the following assumptions:
(i) there are finitely many prime powers $p^e$ with $f(p^e)\mid n$, and we can compute them all; 
(ii) the prime factorization of $n$ is known, otherwise it may be a bottleneck to obtain 
(e.g., Contini et al.~\cite{Contini2006} proved that computing $\varphi^{-1}(\cdot)$ can be used for factoring semiprime integers with little overhead).
In particular, our algorithm is well applicable for computing $\varphi^{-1}(n)$ and $\sigma_k^{-1}(n)$ with $k>0$,\footnote{While $\tau(n)$ can be viewed as a special case of $\sigma_k(n)$ with $k=0$,
the assumption (i) does not hold in this case already for $n=2$.}
which we will use for illustration purposes.
While computing inverses of Euler's totient function $\varphi(\cdot)$ was studied to some extent~\cite{Gupta1981,Contini2006,Coleman2009},
computing inverses of other multiplicative functions such as $\sigma_k(\cdot)$, to the best of our knowledge, was not addressed in the literature.
Our algorithm may be viewed as a generalization and streamlining of the ``intelligent exhaustive search'' for $\varphi^{-1}(n)$ in \cite{Contini2006}.

We present an underlying idea of the algorithm in the elegant form of formal Dirichlet series, which allow us to easily extend it to computing certain functions of the inverses, 
such as their power sums (including cardinality as the 0-th power sum) or extrema, without and faster than direct enumeration of the inverses. For example, our algorithm can establish that
\beq\label{sigma1e1000}
\left|\sigma_1^{-1}(10^{1000})\right| = 15,512,215,160,488,452,125,793,724,066,873,737,608,071,476.
\eeq

\section{Formal Dirichlet series framework}

From now on, we assume that $f$ is a fixed multiplicative function.

We find it convenient to define binary multiplication $\times$ and addition $+$ operations on sets of positive integers as follows:
$U\times V = \{ u\cdot v\ :\ u\in U,\ v\in V\}$ and $U+V = U\cup V$.
Equipped with these operations the set $\P(\mathbb{Z}_{>0})$ of finite\footnote{The combinatorial identities~\eqref{F1}, \eqref{F3}, \eqref{F1X}, \eqref{F3X} 
proved in this section
hold in the case of infinite pre-images as well, while our restriction to the finite case is dictated purely by computational needs.}
subsets of positive integers 
forms a commutative semiring (with the additive identity $\emptyset$ and the multiplicative identity $\{1\}$) and allows us to consider formal Dirichlet series 
with coefficients from this semiring.\footnote{To
emphasize that the defined $+$ and $\times$ are semiring operations,
we use their big versions (in place of more traditional $\Sigma$ and $\prod$) to denote the corresponding series summation and product operators.
Some background information on formal Dirichlet series over semirings is given in the Appendix.}

\begin{theorem} We have the following identity for formal Dirichlet series
of variable $s$ over the semiring $(\P(\mathbb{Z}_{>0}),+,\times)$:
\begin{equation}\label{F1}
\bigplus_{n\geq 1}\enskip \frac{f^{-1}(n)}{n^s} = \bigtimes_{\text{prime}\ p}\enskip \bigplus_{e=0}^{\infty}\enskip \frac{\{p^e\}}{f(p^e)^s}.
\end{equation}
For a fixed positive integer $n$ and every divisor $d\mid n$, we further have
\begin{equation}\label{F3}
f^{-1}(d) = \coeff_{d^{-s}}\enskip \bigtimes_{\text{prime}\ p}\enskip \bigplus_{e:\ f(p^e)\mid n}\enskip \frac{\{p^e\}}{f(p^e)^s}.
\end{equation}
\end{theorem}

\begin{proof}
Multiplicativity of $f$ implies that if $n=f(m)$ and $m=p_1^{e_1}\cdot p_2^{e_2}\cdots p_k^{e_k}$, where $p_1 < p_2 < \dots < p_k$ are primes, 
then $n = f(p_1^{e_1}) \cdot f(p_2^{e_2})\cdots f(p_k^{e_k})$. It follows that $n$ is the product of factors of the form $f(p^e)$, where $p$ is a prime and $e$ is a positive integer,
and no two such factors share the same $p$. In other words,
\begin{equation}\label{Faux}
f^{-1}(n) = \bigplus_{f(p_1^{e_1}) \cdots f(p_k^{e_k})=n}\enskip \bigtimes_{i=1}^k\enskip \{ p_i^{e_i} \},
\end{equation}
where the sum is taken over various tuples of primes $p_1 < p_2 < \dots < p_k$ (with arbitrary $k\geq 0$)
and various positive integer exponents $e_1, e_2, \dots, e_k$ that satisfy $f(p_1^{e_1})\cdot f(p_2^{e_2})\cdots f(p_k^{e_k})=n$.
Multiplying \eqref{Faux} by $n^{-s}$, we get
$$\frac{f^{-1}(n)}{n^s} = \bigplus_{f(p_1^{e_1}) \cdots f(p_k^{e_k})=n}\enskip \bigtimes_{i=1}^k\enskip \frac{\{ p_i^{e_i} \}}{f(p_i^{e_i})^s}.$$
Summing over $n\geq 1$, we obtain 
$$\bigplus_{n\geq 1}\enskip \frac{f^{-1}(n)}{n^s} = \bigplus_{\mycom{p_1 < p_2 < \dots < p_k}{e_1, e_2,\dots,e_k > 0}}\enskip \bigtimes_{i=1}^k\enskip \frac{\{ p_i^{e_i} \}}{f(p_i^{e_i})^s} 
= \bigtimes_{\text{prime}\ p}\enskip \bigplus_{e=0}^{\infty}\enskip \frac{\{p^e\}}{f(p^e)^s},$$
which proves \eqref{F1}.

We remark that terms with $e=0$ in \eqref{F1} represent multiplicative identities (i.e., $\frac{\{p^0\}}{f(p^0)^s}=\frac{\{1\}}{1^s}$ for any prime $p$), 
while for a prime power $p^e$ with $e>0$, $f(p^e)$ may participate in a factorization of $n$ only if $f(p^e)\mid n$. 
Hence, to obtain the full pre-image $f^{-1}(n)$ from \eqref{F1} for a given $n$, 
we can restrict our attention only to such prime powers:
\begin{equation}\label{F2}
f^{-1}(n) = \coeff_{n^{-s}}\enskip \bigtimes_{\text{prime}\ p}\enskip \bigplus_{e:\ f(p^e)\mid n}\enskip \frac{\{p^e\}}{f(p^e)^s}.
\end{equation}
We further remark that for every divisor $d\mid n$, the coefficients of $d^{-s}$ in the series in the right hand side of \eqref{F2} and \eqref{F1} coincide, which implies formula \eqref{F3}.
\end{proof}

Let $(X,\oplus,\otimes)$ be a commutative semiring. 
A mapping $C:(\P(\mathbb{Z}_{>0}),+,\times)\rightarrow (X,\oplus,\otimes)$ is a \emph{weak homomorphism} if for any $U,V\in\P(\mathbb{Z}_{>0})$, we have
$C(U\times V) = C(U)\otimes C(V)$ whenever the sets $U$ and $V$ are element-wise \emph{coprime} (i.e., $\gcd(u,v)=1$ for any $u\in U$ and $v\in V$),
and $C(U+V) = C(U)\oplus C(V)$ whenever $U,V$ are disjoint.
It is easy to see that if $C$ is a homomorphism (i.e., $C(U+V) = C(U)\oplus C(V)$ and $C(U\times V) = C(U)\otimes C(V)$ hold unconditionally), then it is also a weak homomorphism.

\begin{theorem} Let $(X,\oplus,\otimes)$ be a commutative semiring and $C:(\P(\mathbb{Z}_{>0}),+,\times)\rightarrow (X,\oplus,\otimes)$ be a weak homomorphism, then
\begin{equation}\label{F1X}
\bigoplus_{n\geq 1} \frac{C(f^{-1}(n))}{n^s} = \bigotimes_{\text{prime}\ p}\enskip \bigoplus_{e=0}^{\infty}\enskip \frac{C(\{p^e\})}{f(p^e)^s}.
\end{equation}
Furthermore, for a fixed positive integer $n$ and every divisor $d\mid n$,
\begin{equation}\label{F3X}
C(f^{-1}(d)) = \coeff_{d^{-s}}\enskip \bigotimes_{\text{prime}\ p}\enskip \bigoplus_{e:\ f(p^e)\mid n}\enskip \frac{C(\{p^e\})}{f(p^e)^s}.
\end{equation}
\end{theorem}

\begin{proof}
We remark that the sets inside the product in \eqref{Faux} are coprime, while the products inside the sum are disjoint. Since $C$ is a weak homomorphism, we have
$$C(f^{-1}(n)) = \bigoplus_{f(p_1^{e_1}) \cdots f(p_k^{e_k})=n}\enskip \bigotimes_{i=1}^k\enskip C( \{ p_i^{e_i} \}),$$
which further implies identity \eqref{F1X}.
Formula \eqref{F3X} is derived from \eqref{F1X} with the same arguments we used to derive \eqref{F3} from \eqref{F1}.
\end{proof}

Formula \eqref{F3X} under appropriate choice of the weak homomorphism $C$ 
and its codomain $(X,\oplus,\otimes)$ allows us to efficiently compute certain functions of the inverses without their direct enumeration.
In the next section we give some particular examples.

\section{Examples of weak homomorphisms}

Our first, rather trivial example is given by $(X,\oplus,\otimes) = (\P(\mathbb{Z}_{>0}),+,\times)$ with $C$ being the identity homomorphism. 
In this case, formulae \eqref{F1X} and \eqref{F3X} simply represent the original formulae \eqref{F1} and \eqref{F3} for the full pre-images.
We will keep this trivial example in mind to fit computation of full pre-images into our generic algorithm.

Our second example is given by $(X,\oplus,\otimes) = (\mathbb{Z}_{\geq 0},\max,\cdot)$, which is a commutative semiring of nonnegative integers with
a binary maximum operation (i.e., $u\oplus v = \max\{u,v\}$) and the standard integer multiplication. 
The mapping $C(U) = \max \left(U\cup\{0\}\right)$, giving the maximum element of $U\ne\emptyset$ or 0 for $U=\emptyset$, 
represents a homomorphism between $(\P(\mathbb{Z}_{>0}),+,\times)$ and $(\mathbb{Z}_{\geq 0},\max,\cdot)$.

Similarly, the mapping $C(U)=\min \left(U\cup\{\infty\}\right)$ represents a homomorphism between $(\P(\mathbb{Z}_{>0}),+,\times)$ and the commutative semiring 
$(X,\oplus,\otimes) = (\mathbb{Z}_{>0}\cup\{\infty\},\min,\cdot)$,
where $m\otimes \infty=\infty\otimes m=\infty$ and $m \oplus \infty = \infty \oplus m = m$ for any element $m\in X$ (i.e., $\infty$ represents an additive identity).

An example of a weak homomorphism, which is not a homomorphism, is given by 
$(X,\oplus,\otimes) = (\mathbb{Z}_{\geq 0},+,\cdot)$, a semiring of nonnegative integers with the standard integer addition and multiplication, and 
$C_q(U) = \sum_{u\in U} u^q$, where $q$ is a fixed nonnegative integer.
In particular, $C_0(U)=|U|$ represents the cardinality of a set $U$, while $C_1(U)$ is the sum of elements of $U$.

\section{Algorithm for computing $C(f^{-1}(n))$}

In addition to a multiplicative function $f$, we now fix 
a weak homomorphism $C$ from $(\P(\mathbb{Z}_{>0}),+,\times)$ to a commutative semiring $(X,\oplus,\otimes)$.
To compute $C(f^{-1}(n))$ for a given integer $n$ with a known prime factorization, 
we iteratively compute the right hand side of \eqref{F3X} restricted to the terms with denominators $d^s$ for $d\mid n$. 
This computation naturally splits into three major steps outlined below.

\

\textsc{Step 1.} From the prime factorization of $n$, we easily compute the set of its divisors $D$.
Clearly, $|D|=\tau(n)$.

\

\textsc{Step 2.} We compute the \emph{atomic series}
$$L_p = \bigoplus_{e:\ f(p^e)\mid n} \frac{C(\{ p^e \})}{f(p^e)^s} = \bigoplus_{d\in D} \frac{A_d}{d^s}$$
for every prime $p$ that admits at least one\footnote{We remark that if there is no such $e>0$, then $L_p=\frac{C(\{1\})}{1^s}$
represents the identity for $\otimes$-multiplication of formal Dirichlet series.}
integer $e>0$ with $f(p^e)\mid n$.
Here 
\begin{equation}\label{FAd}
A_d = \bigoplus_{e:\ f(p^e)=d} C(\{ p^e \}).
\end{equation}
We remark that finiteness of full pre-images of the function $f$ implies that the number of the atomic series is finite, 
since for each atomic series $L_p$, some positive power of $p$ must belong to the finite set $\bigcup_{d\in D} f^{-1}(d)$.
Internally it is convenient to store each such atomic series $L_p$ as an associative array $d\mapsto A_d$ indexed by elements $d\in D$.

\

\textsc{Step 3.} We multiply the constructed atomic series $L_{p_1}, L_{p_2}, \dots, L_{p_{\ell}}$
and compute partial products $P_0 = \frac{C(\{1\})}{1^s}$, $P_1 = P_0 \otimes_D L_{p_1}$, $\ldots$, $P_{\ell} = P_{\ell-1} \otimes_D L_{p_{\ell}}$,
where $\otimes_D$ denotes the result of $\otimes$ restricted to the terms with denominators $d^s$ for $d\in D$.
Each multiplication is computed with the formula:
\beq\label{otD}
\left(\bigoplus_{d\in D} \frac{A_d}{d^s}\right) \otimes_D 
\left(\bigoplus_{d\in D} \frac{B_d}{d^s}\right) 
= \bigoplus_{d\in D} \frac{\bigoplus_{t\mid d} A_t \otimes B_{d/t}}{d^s}.
\eeq
That is, if the associative arrays $d\mapsto A_d$ and $d\mapsto B_d$ represent $L_{p_{j+1}}$ and the partial product $P_j$,
then we compute the partial product $P_{j+1}$ as an associative array $d\mapsto \bigoplus_{t\mid d} A_t \otimes B_{d/t}$.

For every $d\in D$, the coefficient of $d^{-s}$ in the final product $P_{\ell}$ gives us $C(f^{-1}(d))$. In particular, $C(f^{-1}(n)) = \coeff_{n^{-s}}\enskip P_{\ell}$.

\

While Step 1 of the algorithm is rather trivial and takes $O(\tau(n))$ arithmetic operations on integers of length $O(\log n)$, 
Step 2 is specific to a particular function $f$ and illustrated with some examples in the next section.
Below we analyze the time complexity of the generic Step 3.

\begin{theorem} Let $n$ be an integer and $D$ be the set of divisors of $n$. Given $\ell$ atomic series for $C(f^{-1}(n))$,
their $\otimes_D$-product can be computed with $O(\ell\cdot \tau(n)^2)$ operations in $(X,\oplus,\otimes)$.
\end{theorem}

\begin{proof}
Computation of the $\otimes_D$-product consists of $\ell$ iterative computations of pairwise $\otimes_D$-multiplications defined by \eqref{otD}.
Since each operand of such $\otimes_D$-multiplication contains $O(\tau(n))$ terms, 
its computation takes $O(\tau(n)^2)$ operations $\oplus$ and $\otimes$.
\end{proof}

As we will see in the next section, for Euler's totient function $\varphi(\cdot)$ we have $\ell \leq \tau(n)$.
In particular, Step 3 in computation of the size or extrema of $\varphi^{-1}(n)$ takes $O(\tau(n)^3)$ arithmetic operations on integers of length $O(\log n)$, 
which all can be done in $O(\tau(n)^3\cdot\log^2 n)$ time.
Combining with results of the next section, we obtain that Steps 1-3 in this case can be done in $O(\tau(n)\cdot\log^2 n\cdot(\tau(n)^2 + \log^4 n))$.
Similarly, for the function $\sigma_k(\cdot)$, we have $\ell \leq \tau(n)\cdot\log_2(n)$,
thus computation of the size or extrema of $\sigma_k^{-1}(n)$ can be done in $O(\tau(n)\cdot\log^3 n\cdot(\tau(n)^2 + \log^4 n))$ time.

\section{Computation of atomic series}

As we explained above, our algorithm is generic and works for any multiplicative function $f$, 
provided that we can construct atomic series $L_{p_1}, \dots, L_{p_{\ell}}$. 
To construct them, we need to determine suitable primes $p$ and compute the corresponding series coefficients $A_d$ defined by \eqref{FAd}.
Below we describe such computation in details for the functions $\varphi(\cdot)$ and $\sigma_k(\cdot)$. 
We remark that in both cases we rely on primality testing, 
which for a number with $n$ bits can be done in time $O(\log^{6+\epsilon} n)$ for any $\epsilon>0$~\cite{Lenstra2002}.
Using probabilistic primality test (e.g., Miller--Rabin test~\cite{Rabin1980}) can save a factor of $\log^3 n$.

\subsection{Euler's totient function}

Our goal is to find the prime powers $p^e$ such that $\varphi(p^e)\mid n$.
Since for $e>0$, $\varphi(p^e) = (p-1)p^{e-1}$, the divisibility $\varphi(p^e)\mid n$ implies 
that $p-1$ divides $n$ and $e\leq \nu_p(n)+1$, where $\nu_p(n)$ is the $p$-adic valuation of $n$ (i.e., the maximum integer $t$ such that $p^t$ but not $p^{t+1}$ divides $n$).
So we need to compute the set $S = \{ p : p-1\in D\ \text{and}\ p\ \text{is prime}\}$, which can be done by going over the elements $d$ of $D$ and testing if $p=d+1$ is prime. 
The set $S$ gives us the indices of the atomic series. For every prime $p\in S$, we compute the corresponding atomic series:
$$L_p = \frac{C(\{1\})}{1^s} \oplus \bigoplus_{e=1}^{\nu_p(n)+1} \frac{C(\{p^e\})}{((p-1)p^{e-1})^s}.$$

\

\textit{Example.} For $n=12$ with the set of divisors $D=\{1,2,3,4,6,12\}$, we obtain the set of primes $S=\{2,3,5,7,13\}$. 
If $C(U)$ is a weak homomorphism into $(X,\oplus,\otimes) = (\mathbb{Z}_{\geq 0},+,\cdot)$ giving the sum of the elements of $U$, then the corresponding atomic series are
$$
\begin{array}{lll}
L_2 & = & \frac{1}{1^s} \oplus \frac{2}{1^s} \oplus \frac{4}{2^s} \oplus \frac{8}{4^s} \\
    & = & \frac{3}{1^s} \oplus \frac{4}{2^s} \oplus \frac{8}{4^s},\\
L_3 & = & \frac{1}{1^s} \oplus \frac{3}{2^s} \oplus \frac{9}{6^s},\\
L_5 & = & \frac{1}{1^s} \oplus \frac{5}{4^s},\\
L_7 & = & \frac{1}{1^s} \oplus \frac{7}{6^s},\\
L_{13} & = & \frac{1}{1^s} \oplus \frac{13}{12^s}.
\end{array}
$$
The partial $\otimes_D$-products in this case are $P_0 = \frac{C(\{1\})}{1^s}=\frac{1}{1^s}$ and
$$
\begin{array}{lllll}
P_1 &=& P_0 \otimes_D L_2 &=& \frac{3}{1^s} \oplus \frac{4}{2^s} \oplus \frac{8}{4^s},\\
P_2 &=& P_1 \otimes_D L_3 &=& \frac{3}{1^s} \oplus \frac{13}{2^s} \oplus \frac{20}{4^s} \oplus \frac{27}{6^s} \oplus \frac{36}{12^s},\\
P_3 &=& P_2 \otimes_D L_5 &=& \frac{3}{1^s} \oplus \frac{13}{2^s} \oplus \frac{35}{4^s} \oplus \frac{27}{6^s} \oplus \frac{36}{12^s},\\
P_4 &=& P_3 \otimes_D L_7 &=& \frac{3}{1^s} \oplus \frac{13}{2^s} \oplus \frac{35}{4^s} \oplus \frac{48}{6^s} \oplus \frac{127}{12^s},\\
P_5 &=& P_4 \otimes_D L_{13} &=& \frac{3}{1^s} \oplus \frac{13}{2^s} \oplus \frac{35}{4^s} \oplus \frac{48}{6^s} \oplus \frac{166}{12^s}.
\end{array}
$$
The coefficient of $12^{-s}$ in $P_5$ gives the sum of $\varphi^{-1}(12) = \{ 13, 21, 26, 28, 36, 42\}$.

\

To analyze the running time of this algorithm, let $T_C(m)$ be the maximum time required to compute $C(\{t\})$ for a positive integer $t$ having at most $m$ bits.

\begin{theorem}
Given an integer $n$ and the set of its divisors $D$, the atomic series for $C(\varphi^{-1}(n))$ 
can be computed in time $O(\tau(n)\cdot \log n \cdot (\log^{5+\epsilon} n + T_C(2\log n)))$ for any $\epsilon>0$.
\end{theorem}

\begin{proof}
The proposed algorithm performs $O(\tau(n))$ primality tests of positive integers below $n+1$, each of which takes time $O(\log^{6+\epsilon} n)$.
For every identified prime $p$, it further computes $\nu_p(n)+2=O(\log n)$ values of $C$ on singleton sets with elements below $p^{\nu_p(n)+1}\leq n(n+1)$, 
which takes time $O(\log n\cdot T_C(2\log n))$.
\end{proof}

We remark that for $C$ computing the size or extrema of $\varphi^{-1}(n)$, we have $T_C(m)=O(1)$ so that the time complexity for computing the atomic series 
becomes simply $O(\tau(n)\cdot \log^{6+\epsilon} n)$.

\subsection{Power sum of divisors}

Let $k$ be a positive integer. Our goal is to find the prime powers $p^e$ such that $\sigma_k(p^e)\mid n$, i.e., $\sigma_k(p^e)=d$ for some $d\in D$.
Since $\sigma_k(p^e) = 1 + p^k + p^{2k} + \dots + p^{ek}$, we have $p^{ek} < d \leq (1+p)^{ek}$ or $p^{ek} \leq d-1 < (1+p)^{ek}$, implying that $p=\lfloor (d-1)^{1/(ek)} \rfloor$.
We let $d$ run over $D$ and $e$ run incrementally from $1$ to $\lfloor \frac{\log_2(d-1)}{k} \rfloor$.
For each such pair $(d,e)$, we test whether $p=\lfloor (d-1)^{1/(ek)} \rfloor$ is prime and whether $\frac{p^{(e+1)k}-1}{p^k-1}=d$. 
If both these conditions hold, we have $\sigma_k(p^e)=d$ and add the term $\frac{C(\{p^e\})}{d^s}$ to $L_p$.
Here we assume that initially all $L_p = \frac{C(\{1\})}{1^s}$, and only those $L_p$ that were enriched with additional terms in the above process represent the atomic series.

\

\textit{Example.} Let us compute the minimum of $\sigma_1^{-1}(n)$ for $n=42$. 
So we use the mapping $C(U)=\min \left(U\cup\{\infty\}\right)$ into the semiring $(X,\oplus,\otimes) = (\mathbb{Z}_{>0}\cup\{\infty\},\min,\cdot)$.
We compute the set  $D=\{1, 2, 3, 6, 7, 14, 21, 42\}$ of divisors of $n$ and use it as described above to determine that only prime powers $p^e\in\{2, 2^2, 5, 13, 41\}$ satisfy $\sigma_1(p^e)\mid n$. 
They give rise to the following atomic series:
$$
\begin{array}{lll}
L_2 & = & \frac{1}{1^s} \oplus \frac{2}{3^s} \oplus \frac{4}{7^s},\\
L_5 & = & \frac{1}{1^s} \oplus \frac{5}{6^s},\\
L_{13} & = & \frac{1}{1^s} \oplus \frac{13}{14^s},\\
L_{41} & = & \frac{1}{1^s} \oplus \frac{41}{42^s}.
\end{array}
$$
The partial $\otimes_D$-products in this case are $P_0 = \frac{C(\{1\})}{1^s}=\frac{1}{1^s}$ and
$$
\begin{array}{lllll}
P_1 &=& P_0 \otimes_D L_2 &=& \frac{1}{1^s} \oplus \frac{2}{3^s} \oplus \frac{4}{7^s},\\
P_2 &=& P_1 \otimes_D L_5 &=& \frac{1}{1^s} \oplus \frac{2}{3^s} \oplus \frac{5}{6^s} \oplus \frac{4}{7^s} \oplus \frac{20}{42^s},\\
P_3 &=& P_2 \otimes_D L_{13} &=& \frac{1}{1^s} \oplus \frac{2}{3^s} \oplus \frac{5}{6^s} \oplus \frac{4}{7^s} \oplus \frac{13}{14^s} \oplus \frac{20}{42^s},\\
P_4 &=& P_3 \otimes_D L_{41} &=& P_3.
\end{array}
$$
The coefficient of $42^{-s}$ in $P_4=P_3$ gives the minimum of $\sigma_1^{-1}(42) = \{ 20,26,41 \}$, which is $20$.

\

As above, let $T_C(m)$ be the maximum time required to compute $C(\{t\})$ for a positive integer $t$ having at most $m$ bits.

\begin{theorem}
Given an integer $n$ and the set of its divisors $D$, 
the atomic series for $C(\sigma_k^{-1}(n))$ can be computed in time $O(\tau(n)\cdot \log n \cdot (\log^{6+\epsilon} n + T_C(\log n)))$ for any $\epsilon>0$.
\end{theorem}

\begin{proof}
The proposed algorithm performs $O(\tau(n)\cdot \log n)$ arithmetic operations and primality tests on positive integers below $n$, each of which takes time $O(\log^{6+\epsilon} n)$.
For every identified prime power $p^e$, it further computes $C(\{p^e\})$, which takes time $O(T_C(\log n))$.
\end{proof}

As above, for $C$ computing the size or extrema of $\sigma_k^{-1}(n)$, we have $T_C(m)=O(1)$ so that the time complexity for computing the atomic series 
becomes simply $O(\tau(n)\cdot \log^{7+\epsilon} n)$. 

\section{Examples in the OEIS}

The Online Encyclopedia of Integer Sequences~\cite{OEIS} contains a number of sequences,
for which the proposed algorithm can compute many terms:%
\footnote{Some of the terms in these sequences were computed by Ray Chandler.}

\

\begin{center}
\begin{tabular}{|c||c|c|c|c|c|}
\hline
& $\varphi^{-1}(n!)$ & $\sigma_1^{-1}(n!)$ & $\varphi^{-1}(10^n)$ & $\sigma_1^{-1}(10^n)$ & $\sigma_1^{-1}(p_n\#)$ \\
\hline\hline
size & \seqnum{A055506} & \seqnum{A055486} & \seqnum{A072074} & \seqnum{A110078} & \seqnum{A153078} \\
\hline
min & \seqnum{A055487} & \seqnum{A055488} & \seqnum{A072075} & \seqnum{A110077} & \seqnum{A153076} \\
\hline
max & \seqnum{A165774} & \seqnum{A055489} & \seqnum{A072076} & \seqnum{A110076} & \seqnum{A153077} \\
\hline
\end{tabular}
\end{center}

\

\noindent
In particular, the value of $\left|\sigma_1^{-1}(10^{1000})\right|$ in \eqref{sigma1e1000} represents the $1000$-th term of the sequence \seqnum{A110078}.

\section{Acknowledgements}

The work was supported by the National Science Foundation under Grant No. IIS-1462107.


\section*{Appendix. Formal Dirichlet series}

As we are not aware if Dirichlet series were considered over commutative semirings like $(\P(\mathbb{Z}_{>0}),+,\times)$ before,
we feel obliged to briefly overview their properties.

Let $\X=(X,\oplus,\otimes)$ be a commutative semiring. 
A \emph{formal}\footnote{The term ``formal'' in the description of Dirichlet series and its variable $s$ reflects the fact that they do not take any values 
(thus a formal Dirichlet series is not a function of $s$).
In contrast, conventional Dirichlet series (e.g., Riemann zeta function $\zeta(s)=\sum_{i=1}^{\infty}\frac{1}{i^s}$) are often viewed as functions of a real or complex variable $s$.} 
\emph{Dirichlet series}
over $\X$ is an infinite sequence $(x_1,x_2,\dots)$ of elements of $X$, which we find convenient to represent as
$$\frac{x_1}{1^s} \oplus \frac{x_2}{2^s} \oplus \dots = \bigoplus_{i=1}^{\infty} \frac{x_i}{i^s},$$
where $s$ is a formal variable.

It is easy to check that formal Dirichlet series over $\X$ form a commutative semiring under the addition and multiplication binary operations inherited from $\X$ as follows:
$$\left( \bigoplus_{i=1}^{\infty} \frac{x_i}{i^s} \right) \oplus \left( \bigoplus_{i=1}^{\infty} \frac{y_i}{i^s} \right) = \bigoplus_{i=1}^{\infty} \frac{x_i\oplus y_i}{i^s}$$
and
$$\left( \bigoplus_{i=1}^{\infty} \frac{x_i}{i^s} \right) \otimes \left( \bigoplus_{i=1}^{\infty} \frac{y_i}{i^s} \right) = \bigoplus_{i=1}^{\infty} \frac{\bigoplus_{j\mid i} x_j\otimes y_{i/j}}{i^s}.$$
These operations can be also viewed as an extension of the operations defined on single terms:
$$\frac{x}{m^s} \oplus \frac{y}{m^s} = \frac{x\oplus y}{m^s}$$
and
$$\frac{x}{m^s} \otimes \frac{y}{n^s} = \frac{x\otimes y}{(m\cdot n)^s}.$$

If $\epsilon$ and $\iota$ represent respectively the additive and multiplicative identities in $\X$, then 
$$\mathcal{E} = \bigoplus_{i=1}^{\infty} \frac{\epsilon}{i^s}$$
and 
$$\mathcal{I} = \frac{\iota}{1^s}\oplus \bigoplus_{i=2}^{\infty} \frac{\epsilon}{i^s}$$ 
represent respectively the additive and multiplicative identities in the semiring of formal Dirichlet series over $\X$.
It is often convenient to omit terms with coefficients equal $\epsilon$, e.g., we can simply write $\mathcal{I} = \frac{\iota}{1^s}$.

We denote the coefficient of $\frac{1}{d^s}$ in a formal Dirichlet series $F$ by $\coeff_{d^{-s}}\enskip F$.

\bibliographystyle{jis} 
\bibliography{invphi.bib} 

\bigskip
\hrule
\bigskip

\noindent 2010 {\it Mathematics Subject Classification}:
Primary 11Y16; Secondary 11A05, 11Y55, 11Y70, 30B50.

\noindent \emph{Keywords: }
multiplicative function, Euler's totient function, sum of divisors, inverse function, Dirichlet series

\bigskip
\hrule
\bigskip

\noindent (Concerned with sequences
\seqnum{A055486},
\seqnum{A055487},
\seqnum{A055488},
\seqnum{A055489},
\seqnum{A055506},
\seqnum{A072074},
\seqnum{A072075},
\seqnum{A072076},
\seqnum{A110076},
\seqnum{A110077},
\seqnum{A110078},
\seqnum{A153076},
\seqnum{A153077},
\seqnum{A153078},
\seqnum{A165774}.)

\end{document}